\documentclass[twocolumn]{article}
\pdfoutput=1

\usepackage[margin=1in]{geometry}
\usepackage{amsthm,amsmath,amsfonts,mathtools}
\usepackage{thmtools}
\usepackage{thm-restate}
\usepackage[pagebackref]{hyperref}
\usepackage{graphicx}
\usepackage{caption}
\usepackage{subcaption}
\usepackage{float}
\restylefloat{table}
\usepackage{siunitx}
\usepackage{bbm}
\usepackage{bm}
\usepackage{comment}
\usepackage{xcolor}
\usepackage[shortlabels]{enumitem}
\usepackage{algorithm}
\usepackage{algorithmic}
\usepackage{makecell}
\usepackage{titletoc}
\usepackage{tikz}
\usepackage{array,colortbl,multirow,multicol,booktabs}
\usepackage{microtype}
\usepackage[nameinlink,capitalize]{cleveref}
\usepackage{wrapfig}

\newcommand{\footremember}[2]{%
    \footnote{#2}
    \newcounter{#1}
    \setcounter{#1}{\value{footnote}}%
}
\newcommand{\footrecall}[1]{%
    \footnotemark[\value{#1}]%
}

\hypersetup{
    pdftitle={Shadow Hamiltonian Simulation}, 
    pdfauthor={}, 
    colorlinks=true, 
    linkcolor=blue, 
    citecolor=blue, 
    urlcolor=blue 
}

\renewcommand{\backref}[1]{}

\renewcommand{\backrefalt}[4]{%
\ifcase #1 %
\or 
[p.\ #2]%
\else 
[pp.\ #2]%
\fi}

\newcommand{\ketbra}[1]{| #1 \rangle \! \langle #1 |} 

\newtheorem{theorem}{Theorem}
\newtheorem{lemma}[theorem]{Lemma}

\DeclareMathOperator{\tr}{tr}
\DeclareMathOperator{\im}{Im}

\DeclareMathOperator{\poly}{poly}

\renewcommand{\varepsilon}{\epsilon}

\newcommand{\ket}[1]{|#1\rangle}
\newcommand{\bra}[1]{\langle#1|}
\newcommand{\braket}[2]{\langle#1|#2\rangle}

\def\mg{{\mathfrak g}}

\def\rD{{\rm d}}

\def\cO{\mathcal{O}}

\def\cS{\mathcal{S}}

\def\cV{\mathcal{V}}

\def\Pr{\mathrm{Pr}}

\def\poly{\mathrm{poly}}
\def\mg{{\mathfrak g}}

\def\mg{{\mathfrak g}}

\def\one{{\mathchoice {\rm 1\mskip-4mu l} {\rm 1\mskip-4mu l} {\rm
1\mskip-4.5mu l} {\rm 1\mskip-5mu l}}}
\newcommand{\id}{\one}
\def\dqc1{\textsc{DQC1}}
\newcommand{\psibar}{\smash{\overline{\psi}}}

\newcommand{\sket}[1]{| #1 \rangle}   
\newcommand{\polylog}{\mathrm{polylog}}
\newcommand{\C}{\mathbb{C}}
\newcommand{\barpsi}{\overline{\psi}}

\begin{document}
\title{\bfseries Shadow Hamiltonian Simulation}

\author{\hspace{-9pt}
Rolando D. Somma\footremember{google}{Google Quantum AI, Venice, CA 90291, USA}
~~
Robbie King\footrecall{google} \footremember{caltech}{Computational and Mathematical Sciences, Caltech, Pasadena, CA 91125, USA}
~~
Robin Kothari\footrecall{google}
~~
Thomas O'Brien\footrecall{google}
~~
Ryan Babbush\footrecall{google}
}

\date{}

\maketitle

\begin{abstract}
Simulating quantum dynamics is one of the most important applications of quantum computers. 
Traditional approaches for quantum simulation involve preparing the full evolved state of the system and then measuring some 
physical quantity. Here, we present a different and novel approach to quantum simulation that  uses a compressed quantum state that we call the ``shadow state''. The amplitudes of this shadow state are proportional to the time-dependent expectations of a specific set of operators of interest, and  it evolves according to its own Schr\"odinger equation. This evolution 
can be simulated on a quantum computer efficiently under broad conditions. Applications of this approach to quantum simulation problems include simulating the dynamics of exponentially large systems of free fermions or free bosons, 
the latter example recovering a recent algorithm for simulating exponentially many classical harmonic oscillators. These simulations are hard for classical methods and also for traditional quantum approaches, as preparing the full states would require exponential resources. 
Shadow Hamiltonian simulation can also be extended to simulate expectations of more complex operators such as two-time correlators or Green's functions, and to study the evolution of operators themselves in the Heisenberg picture.
\end{abstract}

\section{Introduction}

Simulating the dynamics of quantum systems, often referred to as quantum simulation or Hamiltonian simulation, is a problem
where quantum computers offer an exponential advantage (cf.~\cite{Llo96,SOGKL02,BAC07,WBH+10,CW12,BCC+15,LC17,LC19,campbell2019random,zlokapa2024hamiltonian}). 
The standard presentation of this problem involves an initial quantum state $\ket{\psi(0)}$ that is easy to prepare on a quantum computer and a Hamiltonian $H$ that describes the interactions in the system. The goal is to produce the state at a later time $t>0$, $\ket{\psi(t)}$, which satisfies the Schr\"odinger equation: $\frac {\rm d} {{\rm d} t}\ket{\psi(t)}=-\mathrm{i}H\ket{\psi(t)}$. In applications, once   $\ket{\psi(t)}$ has been prepared, one is often interested in obtaining some physical property
like the ground state energy of the system or some correlation function. 
Many of these properties can be extracted from a limited set of measurements and expectations, such as those of $k$-local operators (e.g., two-body correlations).

While access to the full state $\ket{\psi(t)}$ allows us to obtain the  expectation of {\em any} observable on this state, for certain problems, it might suffice to prepare a different quantum state that does not encode all the information encoded by $\ket{\psi(t)}$, but instead encodes information about specific properties only.
This could result in significant computational savings, motivating us to
introduce a novel framework for quantum simulation 
that we call ``shadow Hamiltonian simulation''. Rather than storing the full state in quantum memory during the quantum computation, we only encode certain expectations in the amplitudes of a quantum state of smaller dimension, which we refer to as the 
``shadow state''. We will see that the shadow state evolves unitarily under its own Schr\"odinger equation and, under broad conditions, this evolution can be simulated efficiently on a quantum computer, unveiling numerous applications.

To illustrate the power of shadow Hamiltonian simulation, we apply it to simulate the dynamics of free-fermion and free-boson systems, qubit systems, and other systems with a Lie algebraic structure. In the first two cases, by encoding only the expectations of products of fermionic or bosonic operators of constant degree, 
we can use the exponential dimension of 
the shadow state to simulate systems of exponential size. 
These instances can be relevant in, for example, first-quantization simulations, 
where increasing the number of modes, corresponding to the number of points in a discretized grid, allow for better approximations of the states~\cite{gruneis2013explicitlybb,Su2021bb}. 
We can then use shadow states to compute, for example, energies of quantum states associated with fermionic or bosonic modes at any time. This includes expectations of Hamiltonians that are interacting
and other interesting quantities that would require exponential resources using standard Hamiltonian simulation or classical methods. For example,
with standard Hamiltonian simulation, the number of registers needed is linear in either the number of particles or the number of modes, while shadow Hamiltonian simulation does not suffer from this exponential growth in space.
Moreover, using a classical--quantum correspondence (Ehrenfest theorem) we are able to reproduce the results of Ref.~\cite{babbush2023exponential} on simulating dynamics of exponentially many classical oscillators (a BQP-complete problem), which can now be seen as an example of shadow Hamiltonian simulation. These are only a few example applications but the approach can be applied more broadly.

The principles behind shadow Hamiltonian simulation can be generalized in a way that allows encoding the expectations of more complex operators, such as two-time correlators, in the amplitudes of a different quantum state. 
This is convenient for exploring other dynamical properties of the system, yet without requiring the preparation of the full system's state. 
Additionally, while shadow Hamiltonian simulation is  designed to track the time-dependent {\em expectations} of certain operators, it can be easily extended to track the evolution of these operators determined by Heisenberg's equations of motion, without making any reference to a system's state. This enables other efficient quantum computations, such as learning unitary operations, or finding signatures of operator spreading (e.g., light cones) or scrambling~\cite{mi2021information}, for which classical algorithms can take exponential time.

Shadow Hamiltonian simulation relaxes the requirements of standard Hamiltonian simulation in the same way that shadow state tomography relaxes the requirements of full quantum state tomography~\cite{aaronson2018shadow}. 
That is, in the traditional state tomography problem, we are given copies of an unknown state and need to learn a {\em full} classical description of the state, from which we can obtain the expectation of any observable. 
This task requires exponentially many (in the system size) copies of the unknown state in general. In contrast, in shadow tomography, the goal is to learn expectations of a {\em limited} set of observables  on the state, a task that requires much fewer samples than traditional state tomography. 
However, this is only a loose analogy, and we stress that the two settings are distinct: while shadow tomography concerns the problem of learning some classical description of a quantum state, shadow Hamiltonian simulation is related to the problem of evolving expectations in time.

\section{Results}
\label{sec:results}

\subsection*{Formalism}
\label{sec:formalism}

{\bf Shadow states.}
Our results use the notion of a ``shadow state'', which we define formally.
Let $S=\{O_1,\dots,O_M\}$ be a finite set of   operators and suppose we are only interested in  the expectations $\langle O_m \rangle$  with respect to a (potentially mixed) quantum state $\rho$. We would like to store a compressed version of $\rho$ that does not necessarily allow us to reconstruct it, but does allow us to recover certain quantities that depend on the $\langle O_m \rangle$'s only.
We call this the shadow state of $\rho$ with respect to the set $S$, and define it as follows:
    \begin{equation}
    \label{eq:shadowstate}
        \ket{\rho; S} := \frac{1}{\sqrt{A}} \begin{pmatrix}
            \langle O_1 \rangle\cr \vdots \cr  \langle O_M \rangle
        \end{pmatrix}
        = \frac{1}{\sqrt{A}} \sum_{m=1}^M \langle O_m \rangle \ket{m},
    \end{equation}
    where $A = \sum_m \langle O_m\rangle^2$ is for normalization.

For a given set $S$, the same shadow state can correspond to many physical states $\rho$, which occurs when $S$ is not a complete set of operators. 
In this case, $\ket{\rho; S}$ contains less information than $\rho$, in the sense that a classical description of $\rho$ cannot be reconstructed from a classical description of $\ket{\rho; S}$ (whereas $\ket{\rho; S}$ is fully specified by $\rho$ and $S$). 
While $\ket{\rho; S}$ is a state of finite dimension $M$, we do not make any assumptions on the dimension of $\rho$. 
This allows us to consider infinite-dimensional systems as well, such as bosonic systems.

Note that shadow states depend on $\rho$ and $S$ only, and do not make reference to any interactions or Hamiltonian.
In~Supp. Inf., we discuss many properties of shadow states. In particular, in Supp. Note 3 we show that if $S$ is an orthogonal set of operators acting on a finite-dimensional system and $\rho=\ketbra \psi$ is a pure state, then $\ket{\rho;S}$ is a projection or ``shadow'' of $\ket \psi \otimes \overline{\ket \psi}$, where $\overline{\ket \psi}$ is the state whose amplitudes are complex conjugates of the amplitudes of $\ket{\psi}$.

One common task that can be solved with shadow states
is the estimation of an arbitrary linear combination of the expectations $\langle O_m \rangle$, which reduces to an overlap estimation problem. Another common task is that of sampling from a distribution where probabilities are $\Pr(m) \propto |\langle O_m \rangle|^2$, which can be done by
measuring the shadow state in the computational basis.
These tasks are related to other problems in quantum simulation, like estimating energies of a subset of  modes in fermionic and bosonic systems of exponential size, which are studied in this article.

{\bf Shadow Hamiltonian simulation.}
In standard Hamiltonian simulation, 
the full initial state $\rho(0)$ is mapped to the final state $\rho(t)$ by applying the evolution operator $U(t)$. This is determined by the Hamiltonian $H$. Instead, in shadow Hamiltonian simulation, we are additionally given a set of operators $S$, and the goal is to map the initial shadow state $\ket{\rho(0);S}$ to the final shadow state $\ket{\rho(t);S}$. The challenge is to construct this mapping, which is different from $U(t)$.

In general, the vector $\ket{\rho(0);S}$ might not contain enough information to completely determine the vector $\ket{\rho(t);S}$ at a later time. 
For it to be possible to obtain $\ket{\rho(t);S}$ from $\ket{\rho(0);S}$, the expectations of the operators in $S$ at time $t$ should be a function of the initial expectations. This occurs when the Hamiltonian $H$ leaves the space of operators spanned by $S$ invariant, in the sense that the commutation relations between $H$ and the operators in $S$ yield a linear combination of operators in $S$ only.
In this article, we make strong use of this ``invariance property'': We say that $H$ and $S$ satisfy the  invariance property (IP) if and only if
   \begin{equation}
   \label{eq:invarianceproperty}
       [H,O_m]= -\sum_{m'=1}^M h_{mm'} O_{m'} \; , \quad \quad \forall m \in [M] \;,
   \end{equation}
   where $h_{mm'}$ are coefficients.
Let ${H}_S$ be the $M \times M$ matrix of entries $h_{mm'}$. Then, we can write \cref{eq:invarianceproperty} as
\begin{equation} \label{eq:invarianceproperty_matrix}
    [H,{\bf{O}}]= -{H}_S \bf{O} \;,
\end{equation}
where ${\bf{O}} := (O_1,\dots,O_M)^T$ is a vector of operators.

The IP ensures that the shadow Hamiltonian simulation problem is solvable, but does not provide the desired mapping. Our first result is that, if the matrix ${H}_S$ is Hermitian, then the shadow state  evolves according to a Schr\"odinger equation, but with the Hamiltonian ${H}_S$ rather than $H$. 

\begin{theorem} \label{thm:main}
    The shadow state $\ket{\rho(t);S}$ satisfies
    \begin{equation} \label{eq:shadow_schrodinger}
        \frac {\rm d} {{\rm d} t}\ket{\rho(t);S}=-\mathrm{i}{H}_S\ket{\rho(t);S} \; .
    \end{equation}
    Furthermore, if ${H}_S$ is Hermitian,  \cref{eq:shadow_schrodinger} is a Schr\"odinger equation.
\end{theorem}

The proof is a direct consequence of Schr\"odinger's equation applied to $\rho(t)$ and given in~Supp. Note 1. For finite-dimensional systems, ${H}_S$ is Hermitian
 when the operators in $S$ are orthogonal, i.e., 
$\tr(O_m^\dagger O_{m'})=\lambda \delta_{mm'}$, $\lambda >0$; see~Supp. Note 1.
Hence, the result reduces 
the shadow Hamiltonian simulation problem to standard Hamiltonian simulation, which can be performed efficiently on a quantum computer
 under broad conditions on ${H}_S$.

Note that
\cref{thm:main}
applies even if the operators in $S$ are not orthogonal or linearly independent, or even if $\rho(t)$ is infinite-dimensional. Additionally, while we assumed $S$ to be time-independent, the results can be generalized to time-dependent operator sets $S(t)$. For this case, the evolution of $\ket{\rho(t);S(t)}$
is not only induced by ${H}_S$ but also by the operator due to terms of the form $\frac{\partial}{\partial t}O_m(t)$, as they will appear in the time-derivatives of the time-evolved operators under $H$. Nevertheless, in this article we focus on time-independent $S$ for simplicity.

While the IP alone does not imply that ${H}_S$ is Hermitian in general, 
it is sometimes possible to take a particular linear combination of the $O_m$'s to fix this; that is, we can use the property $[H,{ A}{\bf O}]=-{ A} {H}_S{\bf O}$ for an $M \times M$ matrix ${ A}$ and take specific linear combinations of the $O_m$'s to obtain a Hermitian matrix.
This provides some flexibility for shadow Hamiltonian simulation, and we will use this observation for simulating bosonic systems. 
Additionally, since quantum circuits can be interpreted as time evolution with a time-dependent Hamiltonian, shadow Hamiltonian simulation can be easily extended to simulate quantum circuits as well; see~Supp. Note 1.
Note that, even if ${H}_S$ is not Hermitian,
it might be possible
to apply other quantum algorithms for simulating the non-unitary evolution arising from 
\cref{eq:shadow_schrodinger} (cf.~\cite{Ber14,childs2021high}), but
these studies are outside the scope of this article.

\subsection*{Applications}
\label{sec:applications}

 We apply shadow Hamiltonian simulation to various quantum systems, and present problems that can be solved efficiently using this framework.
Subsequently, we demonstrate how similar ideas can be applied to related problems: encoding two-time correlators or expectations of products of operators in a quantum state, 
and encoding a time-dependent operator, subject to Heisenberg's equations of motion, in a quantum state.

{\bf Free-fermion systems.} 
\label{sec:free_fermions}
These systems appear ubiquitously in quantum chemistry and physics. Examples include BCS superconductivity~\cite{takahashi2005thermodynamics}, integer quantum Hall effect~\cite{tong2016quantum}, spin models like the transverse-field Ising model~\cite{barouch1971statistical},
and other cases described by mean-field theories. 
For $n$ fermionic modes, the number of degree-$k$ fermionic operators is $\cO(n^k)$, and in this case we could use shadow states to encode the expectations of operators acting on systems that are exponentially large, e.g., $n=2^r$. 
Our goal is to provide a method for shadow Hamiltonian simulation that is efficient,   of complexity polynomial in $r$ or $\polylog(n)$, and polynomial in the evolution time $t$. We will show this is possible for  free-fermion systems under broad conditions.

In terms of Majorana operators,
the Hamiltonian of a free-fermion system is 
$H = \sum_{j,k=1}^{2n} \gamma_{jk} c_j c_k$,
where the interaction strengths $\gamma_{jk}\in \mathbb{C}$
define a $2n \times 2n$ Hermitian matrix ${ \Gamma}$.
Majorana operators are Hermitian and satisfy 
the commutation relations $[c_jc_k,c_l]=2\delta_{lk} c_j - 2\delta_{lj} c_k$. 
These imply that $H$ and the set $S$ of degree-$k$
fermionic operators satisfy the IP for any $1 \le k \le 2n$.
Consider, for example, $S=\{c_jc_k\}_{1 \le j < k \le 2n}$, and note that the $c_jc_k$'s are readily orthogonal.
The amplitudes of the shadow state $\ket{\rho(t);S}$ are proportional to the expectations $\langle c_j c_k(t) \rangle=\tr(\rho(t) \,c_j c_k)$, which in this case coincide with the entries
of the so-called 1-RDM (reduced density matrix) of the evolved system's state $\rho(t)$~\cite{Rubin2018}.
To prepare the shadow state at time $t$, we would like to simulate the Hamiltonian ${H}_S$ in \cref{eq:shadow_schrodinger} efficiently
and also prepare the initial shadow state efficiently.
For these to occur, it suffices to have efficient access to the sparse matrix ${ \Gamma}$ and for $\rho(0)$ to be the free vacuum
(the state with no bare or free fermions) or a fermionic Gaussian state (Slater determinant).
We give more details in~\Cref{sec:methods}
and Supp. Note 2, which
describes the case of free fermions in depth.

Many quantities can be extracted from the shadow state of an exponentially large fermionic system, like
the energy $\langle H_J \rangle$ of a subset $J \subseteq [2n]$
of Majorana modes, where $H_J:=\sum_{(j,k) \in J} \gamma_{jk} c_j c_k$. This is an overlap estimation problem, since
   $\langle H_J \rangle = G \sqrt A \langle \psi_J \ket{\rho ; S}$, 
where $\ket{\psi_J}:=\frac 1 G \sum_{(j,k) \in J} \gamma_{jk} \ket{j,k}$  is a unit quantum state,
$\ket{\rho;S}=\sum_{j,k} \langle c_jc_k \rangle \ket {j,k} / \sqrt{A}$, 
and $G =[\sum_{(j,k)\in J}\gamma_{jk}^2]^{1/2}$ and $A$ are for normalization.  It can be shown that $A \le n$.
When the set $J$ is extensive so that $|J| \propto n$ and ${ \Gamma}$ is sparse,  $G=\cO(\sqrt n \|{ \Gamma}\|_{\rm max})$. 
Then, $\langle H_J \rangle/(n \|{ \Gamma}\|_{\rm max})$ 
is an intensive property that can be computed efficiently
from the overlap $\langle \psi_J \ket{\rho ; S}$ 
using known methods (cf.~\cite{KOS07}).
These problems can also be shown to be BQP-complete.
More generally, we can consider shadow states that encode 
the expectations of degree-$k$ fermionic operators, allowing us to estimate energies of interacting fermionic systems (non-quadratic Hamiltonians) in fermionic Gaussian states and other states.

{\bf Free-boson systems.}
These systems are also prevalent in physics, and describe physical
phenomena such as  superfluidity~\cite{bogoliubov1947theory} and quantum optics~\cite{scully1997quantum}. 
Free-boson systems can also be understood as a collection 
of coupled quantum harmonic oscillators, which are quantum particles evolving under the influence of quadratic potentials. Like in the prior example, we are interested in simulating exponentially large systems efficiently, in time $\polylog(n)$ and $\poly(t)$.

We consider a presentation of a free-boson system given by a quadratic Hamiltonian $H=\frac 1 2 {\bf Y}^T { \Gamma} {\bf Y}$, where ${\bf { Y}}:=(P_1,P_2,\ldots,Q_1,Q_2,\ldots)^T$ is a vector of $2n$ Hermitian operators and $ \Gamma$ is a $2n \times 2n$ Hermitian matrix of interaction strengths (assumed to be real). The operators $P_j$ and $Q_j$ are associated with generalized coordinates like  momentum and position,
and satisfy the cannonical commutation relations
$[Y_j,Y_k]=\mathrm{i} \Omega_{jk}$, for $j,k \in [2n]$. Here,
  $\Omega_{jk}$ is the entry of the symplectic matrix ${ \Omega}=\begin{pmatrix} {\bf 0}& -\one_n \cr \one_n & {\bf 0}\end{pmatrix}$, where $\one_n$ is the $n \times n$ identity. 
 This implies
   $[Y_j Y_k,Y_l]=\mathrm{i} \Omega_{kl} Y_j  + \mathrm{i} \Omega_{jl }Y_k$.
Then, the set of degree-$k$  bosonic operators also satisfies the IP for any $k\ge 1$.

As a special case we could choose $S=\{P_1,P_2,\ldots,Q_1,Q_2,\ldots\}$. However, one immediately arrives at the issue that ${H}_S$ is not Hermitian in general, unless it is number conserving, a case recently addressed in Ref.~\cite{barthe2024gate} within a related context.
To circumvent this issue, we use a decomposition ${ \Gamma}={ \Gamma}_1{ \Gamma}_2$ for some ${ \Gamma}_1$ of dimension $2n \times M$ and ${ \Gamma}_2$ of dimension $M \times 2n$, where $M \ge \mathrm{rank}( \Gamma)$. 
Transforming
${\bf {Y}} \mapsto   {\bf {O}}:={ \Gamma}_2 \bf {Y}$
changes the coefficients in the commutators such that the IP reads
$[H,{\bf {O}}]=-\mathrm{i} { \Gamma}_2 { \Omega} { \Gamma}_1 \bf{O}$.
The set of operators is now $S:=\{O_1,\ldots,O_M\}$; this also satisfies the IP
and we are interested in instances where $\mathrm{i} { \Gamma}_2 { \Omega} { \Gamma}_1 \equiv {H}_S$ is Hermitian. This occurs for Hamiltonians 
of exponentially many coupled quantum harmonic oscillators, e.g.,
\begin{align}
\label{eq:quantumharmonicH}
    H=  \sum_j \frac{(P_j)^2}{2m_j}+ \frac {\kappa_{jj}} 2  (Q_j)^2 +\sum_{k>j} \frac{\kappa_{jk}} 2 (Q_j-Q_k)^2 \;.
\end{align}
Here, $m_j>0$ and $\kappa_{jk} \ge 0$ for all $j \in [n]$, $k \in [n]$, are the masses and spring constants.
For this $H$, the matrix of interaction strengths satisfies ${ \Gamma} \succeq 0$. 
Considering a factorization ${ \Gamma}={ B}{ B}^\dagger$, for some matrix ${ B}$, we can choose ${ \Gamma}_1={ B}$ and ${ \Gamma}_2={ B}^\dagger$, and prove that ${H}_S$ is now Hermitian.

The Hamiltonian ${H}_S$ can be simulated efficiently when we have efficient access to the sparse matrix ${ B}$.
This is discussed
in Ref.~\cite{babbush2023exponential} in great detail, which gives a quantum algorithm for simulating {\em classical} oscillators; both problems are equivalent due to a classical-quantum correspondence.  Also, the initial shadow state can be prepared efficiently in some cases. 
For illustration, consider the case of a bosonic product state $\rho$ where
$\langle Q_j(0) \rangle=0$ and $\frac 1 {\sqrt{ m_j}}\langle P_j (0)\rangle \ne 0$ is the same constant for all $j \in [n]$. 
Then, $\ket{\rho(0);S}=\frac 1 {\sqrt n}\sum_{j=1}^n \ket j \ket j$
is  a simple superposition state that can be prepared with $\cO(\log(n))$ elementary gates. Like in the case of free-fermion systems, other shadow states
can be considered by building upon this example.
Hence, shadow Hamiltonian simulation enables the efficient simulation of the dynamics of exponentially large bosonic systems in many cases.

Shadow states can be used to compute properties
of these systems efficiently. One example is estimating
a semiclassical approximation to the (rescaled) kinetic or potential energies
of a subset of masses or springs.
This is essentially the same problem discussed in Ref.~\cite{babbush2023exponential}, which is BQP-complete, and the reason why this is a semiclassical approximation here is because the expectation of a quadratic bosonic operator is not equal to the product of expectations of single operators in general. 
To go beyond this approximation, we could consider the set $S=\{O_m O_{m'}\}_{1 \le m \le M,1\le m' \le M}$, which also satisfies the IP. With this choice we have
    $\ket{\rho;S} =\frac 1 {\sqrt A} \sum_{m,m'} \langle O_m O_{m'}\rangle \ket{m,m'}$, where $A>0$ is for normalization.
It can be shown that the shadow state at time $t$ can
be obtained from the initial one by evolving with
 ${H}_S \otimes \one_M + \one_M \otimes {H}_S$ for time $t$.
 (\Cref{thm:main2} provides a more general result.) 
The energy of a subset  $I \subseteq [M]$ of quantum oscillators
is $\langle H_I\rangle=\frac 1 2 \sum_{m \in I}\langle (O_m)^2 \rangle$. 
This is also an overlap estimation problem, since
$\langle H_I\rangle =\frac 1 2\sqrt {A |I|} \langle \psi_I \ket{\rho;S}$,
where $\ket{\psi_I}=\frac 1 {\sqrt {|I|}}\sum_{m \in I}\ket{m,m}$ is a unit quantum state.
The normalization constant $A$ can be $\Omega(n^2)$ in general, since for sparse ${ \Gamma}$ the number of operators in $S$ is $\Theta(n^2)$,  and all their expectations can be nonzero. 
However, for some initial states $\rho$ that are of low energy, including some product states, we can prove that most expectations are zero and  $A=\cO( n)$. 
Then, if $|I| \propto n$, the expectation $\langle H_I\rangle/n$ is an intensive property that can be obtained from the overlap $\langle \psi_I \ket{\rho;S}$. 
More generally, if $S$ contains bosonic operators of degree $k>2$,we can perform other calculations like estimating the energies of interacting bosonic systems (non-quadratic Hamiltonians) in the corresponding states.

{\bf Qubit systems.} \label{sec:qubits}
The simplest example, which could be viewed as ``free qubits'', is when the Hamiltonian $H$ is a sum of local Pauli operators each acting non-trivially on only one qubit, i.e., $S=\{X_1,Y_1,Z_1,\ldots,X_n,Y_n,Z_n\}$. 
The IP is satisfied
and ${H}_S$ is Hermitian due to the orthogonality of Pauli operators. The dimension of the shadow state is only $M=3n$ and we can encode the expectations of exponentially many qubits into a shadow state of exponential dimension,
 which uses only $\cO(\log(n))$ many qubits. 
 In this example the qubits do not interact, and hence the system might not be very interesting. 
A different example is when $S$ contains all $n$-qubit Pauli operators. Specifically, for $i,j\in \{0,1\}^n$, let $P_{ij} := X_1^{i_1}Z_1^{j_1} \otimes \cdots \otimes X_n^{i_n}Z_n^{j_n}$ and let $S:=\{P_{ij}:i,j\in \{0,1\}^n\}$. Pauli operators are orthogonal, implying Hermiticity of ${H}_S$, and $|S|=M=4^n$. The shadow state is $\ket{\rho;S} = \frac 1 {\sqrt A}\sum_{i,j\in \{0,1\}^n} \langle P_{ij} \rangle \ket{i,j }$, 
where $\langle P_{ij} \rangle =\tr (\rho P_{ij})$ is the expectation. 
The normalization constant is $A=\sum_{i,j}|\langle P_{ij}\rangle|^2 = 2^n \tr (\rho^2)$, being proportional to the purity of $\rho$.

Let $V_S$ be the ``Bell rotation'', which maps between the Bell basis and the standard basis; see \Cref{fig:Bellsampling}. 
 As we show in~Supp. Note 3, the shadow state of a pure state $\rho=\ketbra \psi$ for this $S$ is    
 $\ket{\rho;S} = V_S \ket{\psi} \otimes \ket{\barpsi}$. We can then prepare the shadow state efficiently whenever we have an efficient circuit to prepare $\ket{\psi}$. (This also gives an efficient circuit for $\ket{\barpsi}$ by complex conjugating all the elementary gates in the circuit for preparing $\ket{\psi}$.) Similarly, it is possible to perform shadow Hamiltonian simulation since, for any unitary $U$, the shadow state corresponding to  $U\ket{\psi}$ is simply $U \ket{\psi} \otimes \overline{U \ket \psi}$.
\begin{figure}[htb]\centering
\includegraphics[scale=.35]{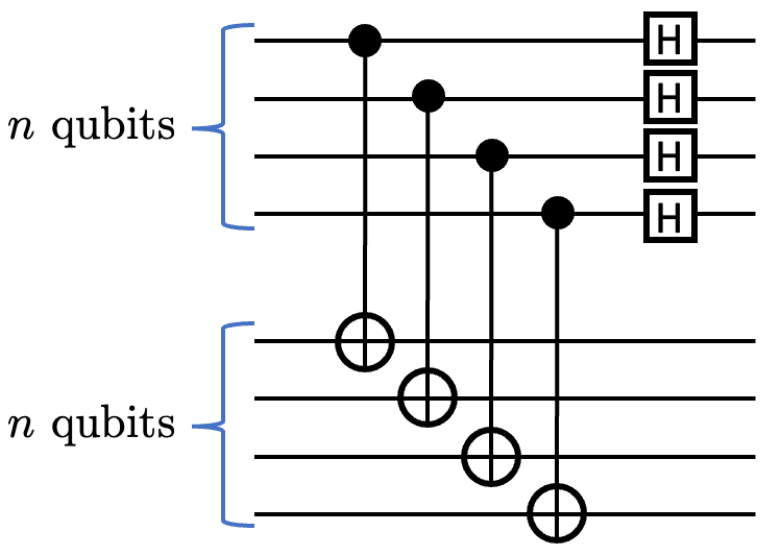} 
\caption{The operation $V_S$ for a system of $n=4$ qubits. The two-qubit gates are CNOTs and H are Hadamard. When acting on $n$ Bell pairs, these gates output a state in the computational basis.}
\label{fig:Bellsampling}
\end{figure}

{\bf Lie-algebraic systems.} \label{sec:lie_algebras}
The case of Lie algebras is a natural one to consider 
in shadow Hamiltonian simulation because the IP is automatically satisfied. More precisely,
a \emph{Lie algebra} $\mathfrak{g}$ is a vector space with a \emph{Lie bracket} binary operation $[\cdot,\cdot] : \mathfrak{g} \times \mathfrak{g} \rightarrow \mathfrak{g}$.
In our setting, the Lie bracket is simply the commutator between operators or matrices: $[x,y]=xy-yx$. Hence, if the Hamiltonian $H$ belongs to a Lie algebra $\mathfrak{g}$, this  implies that any basis of $\mathfrak{g}$ can be chosen as the set $S$ of operators that satisfies the IP.
Indeed, we have seen some examples already: the set of quadratic fermionic operators span the Lie algebra $\mathfrak{so}(2n)$, quadratic bosonic operators span the symplectic Lie algebra $\mathfrak{sp}(2n)$, and all Pauli operators span the Lie algebra $\mathfrak{u}(2^n)$. 
Other examples are spin systems with associated Lie algebras $\mathfrak{su}(2)$ or, more generally, $\mathfrak{su}(n)$.
  This provides a unifying framework for a broad class of situations where shadow Hamiltonian simulation is applicable.
Specifically, when the dimension of the Lie algebra is polynomial in the system size, shadow Hamiltonian simulation is able to tackle quantum systems that are of exponential size.

Let $S=\{O_1,\ldots,O_M\}$ be the operator set and assume it spans the Lie algebra $\mg$.
To apply shadow Hamiltonian simulation, we need to construct an ${H}_S$ that is also Hermitian. This ${H}_S$ is determined by
the structure constants of $\mg$, which are the $f_{jkl}$ in the commutations $ [O_j,O_k] = \sum_{l=1}^M f_{jkl} O_l$. Hermiticity of${H}_S$ can follow from the orthogonality of the $O_j$'s in finite-dimensional systems. Two assumptions ensure that ${H}_S$ can be efficiently simulated: i) having efficient sparse access to the structure constants, and ii) only a few ($\polylog(M)$) terms in $H$ do not commute with a given $O_j$ and these terms can be efficiently computed. 
Additionally, many initial shadow states can be efficiently prepared. Examples are those of lowest-weight states $\ket{{\rm lw}}$, generalizing the free vacuum in the fermionic case,  and
those corresponding to ``generalized coherent'' states of $\mg$. These are obtained as $e^{\mg} \ket{{\rm lw}}$, where $e^{\mg}$ is a unitary in the group induced by $\mg$, and generalize the Gaussian fermionic states discussed when $\mg \equiv \mathfrak{so}(2n)$~\cite{Note1}. We give more details in~\cref{sec:methods}.

{\bf Green's functions and other correlators.}
Shadow Hamiltonian simulation can be extended to encode expectations of products of different time-dependent operators in the amplitudes of a quantum state. 
This is most relevant in the Heisenberg picture, where the time-dependent expectation of an operator $O$ on the state $\rho(t)$ is alternatively described as the expectation of a time-dependent operator $O(t)$ on the initial state $\rho(0)$. The evolution of the operator is given by Heisenberg's equations of motion. 
More generally, we can consider products of operators at different times in the Heisenberg picture, e.g. $O_1(t) O_2(t')$. Let $H$ and $S$ satisfy the IP and define the pure state
\begin{align}
    \ket{\rho;S (t,t')} = \frac{1}{\sqrt{A}} \begin{pmatrix}
          \langle O_1(t) O_1(t') \rangle\cr  \langle O_1(t) O_2(t') \rangle\cr \vdots \cr  
          \langle O_M(t)  O_{M-1}(t')\rangle \cr  
          \langle O_M(t)  O_M(t')\rangle 
    \end{pmatrix} ,
\end{align}
where $\langle O_m(t) O_{m'}(t')\rangle$ are the expectations with respect to the initial state $\rho$, and $A>0$ is for normalization.
The following result, shown in~Supp. Note 1,
extends shadow Hamiltonian simulation 
to this scenario. 

\begin{theorem}
\label{thm:main2}
The state  $\ket{\rho;S (t,t')}$ satisfies
\begin{align}
\label{eq:shadowproduct_schrodinger}
    \frac{\partial}{\partial t}\ket{\rho;S(t,t')} &= -\mathrm{i} ({H}_S \otimes \one_M) \ket{\rho;S(t,t')}\;, \\
    \frac{\partial}{\partial t'}\ket{\rho;S(t,t')} &= -\mathrm{i} (\one_M \otimes {H}_S) \ket{\rho;S(t,t')} \;.
\end{align}
\end{theorem}

We can then encode two-time correlators and Green's functions in the amplitudes of a quantum state. To prepare $\ket{\rho;S(t,t')}$ from $\ket{\rho;S(0,0)}$ efficiently on a quantum computer, it suffices for ${H}_S$ to be Hermitian and to have efficient sparse access to ${H}_S$. The procedure involves evolving with ${H}_S$ acting on the first register for time $t$ and evolving with the same ${H}_S$ acting on the second register for time $t'$.
In addition, it is often possible to prepare the initial state $\ket{\rho;S(0,0)}$ efficiently, if the set of initial expectations contains structure that can be expressed succinctly, similar to the examples discussed above.

Importantly, \cref{thm:main2} can be generalized in two directions: i) we can encode higher order correlation functions $\langle O_{m_1}(t_1) \dots O_{m_q}(t_q) \rangle$, where all the $O_{m_j}$'s are in the same operator set $S$ and evolve under $H$ for different times $t_1,\ldots,t_q$, and ii) when the operators in the product belong to different operator sets $S_1,\ldots,S_q$, each evolving with their own distinct Hamiltonians $H_1,\ldots,H_q$, as long as the corresponding IPs between $S_j$ and $H_j$ are satisfied for all $j$.

{\bf Operators in the Heisenberg picture.}
\label{sec:scrambling}
In a different extension, we can encode a representation of the time-evolved operators themselves in the Heisenberg picture, rather than their expectations, without making reference to the state of the system. 
Consider the operator set $S=\{O_1,\ldots,O_M\}$ and let $Z:=\sum_{m=1}^M z_m O_m$, where $z_m \in \mathbb C$ for all $m \in [M]$.
The time-evolved version in the Heisenberg picture is $Z(t)$ and, if the Hamiltonian $H$ is time-independent, Heisenberg's equation of motion is $\frac{\rD}{\rD t}Z(t) = {\rm i}[H,Z(t)]$.
If $H$ and $S$ further satisfy the IP, we can express $Z(t)=\sum_{m=1}^M z_m(t) O_m$, where $z_m(t) \in \mathbb C$. These coefficients evolve according to
\begin{align}
\label{eq:schrodingercoeff}
  \frac{\rD}{\rD t}z_m(t)  = -{\rm i} \sum_{m'=1}^M h_{m'm} z_{m'}(t) \;,
\end{align}
which is a direct consequence of Heisenberg's equation of motion. The coefficients $h_{m'm}$ are the entries of the matrix ${H}_S$. If this is Hermitian, \cref{eq:schrodingercoeff} is also a Schr\"odinger equation, and it is natural to consider the operator-vector mapping $Z(t) \mapsto \ket{Z(t)} \propto \sum_{m=1}^M z_m(t) \ket m$. In this way, we can encode the time-evolved operator in a quantum state, and  \cref{eq:schrodingercoeff} implies
\begin{align}
  \frac{\rD}{\rD t}\ket{Z(t)}  = -{\rm i} \overline{{H}_S} \ket{Z(t)} \;.    
\end{align}
A quantum algorithm that prepares $\ket{Z(t)}$ involves  preparing the initial state $\ket{Z(0)}$ and performing time evolution with the Hamiltonian  $\overline{{H}_S}$ for time $t$. The algorithm is efficient when these two steps can be implemented efficiently. The analysis can be extended to time-dependent Hamiltonians. In that case, it can be shown that the 
transformation that takes $\ket{Z(0)}$ to $\ket{Z(t)}$
is the evolution operator with the Hamiltonian $\overline{{H}_S(t-s)}$, where the final time $t$ is fixed and $s$ is increased from $0$ to $t$. This resembles the evolution of the operator in the Heisenberg picture that evolves backwards in time. More details are in~Supp. Note 1, where we discuss the case of quantum circuits as well.

The state encoding a time-evolved operator can let us probe many interesting properties of the system. One example is operator growth, which relates to the phenomenon of \emph{quantum scrambling} \cite{schuster2023operator,swingle2016measuring,landsman2019verified,joshi2020quantum,blok2021quantum,mi2021information}. Suppose for simplicity that the set $S$ consists of all Pauli operators $P_{ij}$ on $n$ qubits. We can define an observable $W$ which is diagonal in the computational basis; on basis element $\ket{i,j}$, it takes value equal to the Hamming weight of the operator $P_{ij}$, which counts the number of non-trivial Paulis and we denote $\text{wt}(P_{ij})$. Measuring $\langle Z(t)|W|Z(t)\rangle$ will compute the expected weight at time $t$ given by
\begin{equation}
    \langle Z(t)|W|Z(t)\rangle = \sum_{i,j} z_{ij}(t)^2 \text{wt}(P_{ij})
\end{equation}
Similar to~\cref{thm:main2}, we can also consider the evolution of products of operators and two-time correlators by encoding them in quantum states.

The approach can also be applied for learning certain unitary oracles.
For example, if $S$ is again the set of all Pauli operators $P_{ij}$ on $n$ qubits and the unitary that transforms $S$ is a Clifford operation $C$, we can use this approach to study the evolution of the local Pauli operations $X_1,Z_1,\ldots,X_n,Z_n$. Each of these can be encoded in a corresponding $2n$-qubit state, of a single amplitude 1 in a basis state corresponding to the local Pauli, and all other amplitudes are zero. The corresponding transformation on the $2n$-qubit state will simply be a permutation, since   $C$ maps  Pauli operators to other Pauli operators, i.e., it permutes the elements of $S$. A measurement after the evolution gives the transformation for the local Pauli operator, given as a sequence of $2n$ bits that specify all the $n$ local Pauli operators in the Pauli product. Hence, with $2n$ experiments of this form, we obtain full knowledge of $C$. This requires $2n$ uses of each $C$ and $\bar C$, the complex conjugate. A similar approach is described in Ref.~\cite{low2009learning}. The analysis can be directly generalized to efficiently learn 
unitary oracles that map Pauli products to linear combinations of polynomially many Pauli products and beyond. Other examples
are learning unitary oracles corresponding to evolutions of free-fermion or free-boson models, as long as the evolved (creation or annihilation) operators are mapped to linear combinations involving polynomially many operators. For all these cases, we only require that state tomography 
of the corresponding $\ket{Z(t)}$'s can be performed efficiently.

\subsection*{Quantum speedups in realistic situations}

It was shown in Ref.~\cite{babbush2023exponential} that simulating exponentially-many classical harmonic oscillators is a BQP-hard problem. Since we recover this algorithm using shadow Hamiltonian simulation, this demonstrates that shadow Hamiltonian simulation is capable of providing an exponential quantum speedup in principle. In Ref.~\cite{babbush2023exponential} they also give an exponential oracle separation for the problem of simulating dynamics of coupled classical harmonic oscillators, which can also be solved using shadow Hamiltonian simulation. These systems involve ``long-range'' interactions, in which oscillators that are exponentially far apart can be coupled with a spring.

However, one can ask about quantum speedups in realistic physical geometries. In this case, properties like locality of interactions could be exploited, potentially enabling more efficient classical simulations of the dynamics.
For example, 
suppose we have a system of harmonic oscillators or free fermions on a lattice in $D=3$ spatial dimensions. Simulating time evolution within shadow Hamiltonian simulation has cost almost linear in the evolution time $t$ and polylogarithmic in system size. Classically, one must only simulate the ``lightcone'' of the region of interest, with spacetime volume scaling as $ \sim t^4$. Thus we achieve a \emph{quartic} quantum speedup over classical algorithms such as finite-difference time domain (FDTD),
which 
encodes the shadow state in a vector and simulates
time evolution via matrix-vector multiplication.
More generally, for lattices in $D$ spatial dimensions, the size of the lightcone is $\cO(t^D)$ and classical algorithms can have complexity $\cO(t^{1+D})$,
while the complexity of shadow Hamiltonian simulation is $\tilde \cO(t)$.

\section{Discussions}
\label{sec:conclusions}

We devised a novel approach to quantum simulation
that uses compressed versions of the full quantum states.
This shadow Hamiltonian simulation approach 
allows us to encode limited information about a system's quantum state, like the expectations of a relevant set of operators at any time, in the amplitudes of a different ``shadow state''. We can then exploit the exponential dimension of the shadow state and encode, for example, 
the expectations of operators acting on systems of exponential size.
By doing so we unveiled applications to  quantum systems like fermions and bosons  with an exponential number of modes. For these, we described computations that can be carried efficiently via shadow Hamiltonian simulation, while these computations are hard for classical or traditional quantum approaches, as preparing the full states would require exponential resources. Note that while free-boson and free-fermion systems are integrable and efficient classical solutions are possible, that efficiency is lost when analyzing systems that are exponentially large, as considered in this work.

Additionally, we demonstrated how the ideas underlying shadow Hamiltonian simulation can be applied to two other problems. In one problem, we described how to encode the expectations of two-time correlators or other products of time-dependent operators in a different quantum state. In the other problem we described an approach for encoding operators in the Heisenberg picture in quantum states, without referencing to the state of the system. These developments enable the solution to other problems; an example is studying the growth of the light cone associated with an evolving operator.

Shadow Hamiltonian simulation also shares features of prior results and can improve them, further highlighting the potential of our approach. For example, first-quantization methods enable the simulation of exponentially many fermionic modes as well.  However, known approaches that use first quantization (cf.,~\cite{Su2021bb}) require one register per fermion, and the number of registers can be exponentially large when the number of fermions is exponential (e.g., at half filling).
In contrast, shadow Hamiltonian simulation uses a number of registers that depends on the set of operators $S$, and {\em not} on the number of fermions. 
While it might be possible to adapt first-quantized methods to encode partial information on the fermionic state using fewer registers (e.g., $k$ registers to encode the $k$-RDM), the corresponding fermionic states would often be highly mixed and
our approach based on pure shadow states 
still appears more powerful. For example, it is possible to prove an exponential separation (in the number of modes) in sample complexity between the cases where access to copies of the shadow states is given and where copies of the mixed states that encode the $k$-RDM is given. Additionally, shadow Hamiltonian simulation can deal with free-fermion Hamiltonians that are \emph{not} number preserving, which is not the case of known first-quantization methods. This is useful for studying
models like, for example, the BCS Hamiltonian that appears in superconductivity.
Nevertheless, a main feature of first-quantization methods is that they can simulate interacting fermionic Hamiltonians, but their scaling can be polynomial in the number of modes.

Also, there is an inherent relation between the quantum algorithm of
Ref.~\cite{babbush2023exponential} for simulating exponentially many coupled classical oscillators and shadow Hamiltonian simulation. 
That quantum algorithm shares similar features in that the coordinates (position and momentum) of the oscillators are also encoded in the amplitudes of a quantum state, which can be shown to evolve unitarily. While those are classical variables,  their dynamics   is similar to the dynamics of the expectations of the corresponding operators of exponentially many coupled {quantum} oscillators.  This is a consequence of Ehrenfest's theorem, which relates classical and quantum evolution, applied to quadratic free-boson systems. 
Nevertheless, with shadow Hamiltonian simulation
we can perform more complex calculations that involve expectations of higher-order correlations in bosonic systems, going beyond
Ref.~\cite{babbush2023exponential}.

We anticipate several future directions. 
For instance, we mentioned that quantum algorithms for differential equations might be used when the IP  is satisfied but ${H}_S$ is not Hermitian, and understanding when such algorithms are efficient would be important. 
Another potential direction is the simulation of open quantum system dynamics, and whether shadow states can be used to perform this simulation more efficiently. Also, many
classical systems have an underlying Lie algebra structure, and it would be interesting to unveil other examples whose classical dynamics can be efficiently simulated on a quantum computer using shadow states. 
Furthermore, it might be possible to address interacting systems where the IP is not satisfied by considering products of operators in $S$ up to some degree (cutoff), and attempt an approximate simulation using shadow states.

\section{Methods}
\label{sec:methods}

For shadow Hamiltonian simulation to be efficient, the simulation of ${H}_S$ has to be implemented efficiently and the initial shadow state $\ket{\rho(0);S}$ has to be prepared efficiently. These take place in many interesting cases, and we provide the conditions and methods for efficiency in the examples analyzed.

\subsection*{Efficient Hamiltonian simulation}

Shadow Hamiltonian simulation involves 
the simulation of ${H}_S$ for time $t$ and, to this end, we can use known techniques for standard Hamiltonian simulation like quantum signal processing~\cite{LC17}.
In general, we will consider cases where the number of distinct $O_m$'s in $S$ is exponentially large in some problem size, or where the Hamiltonian $H$ is also a sum of exponentially many such operators. Since ${H}_S$
is determined from $H$ and $S$, specifying it and simulating it in general will require exponential complexity. Nevertheless, as we are interested in problem instances that can be simulated efficiently, we will consider those with a succinct presentation for ${H}_S$ that allows for efficient simulation. By efficient simulation we mean an algorithm of complexity $\poly(t)$ and $\polylog(M)$, and also polynomial in some norm of ${H}_S$.

Quantum signal processing and related efficient approaches for standard Hamiltonian simulation
assume access to a ``block-encoding'' of the Hamiltonian.
This is a unitary operator acting on an enlarged space that contains the Hermitian operator 
${H}_S/\Lambda$ in one of its blocks. The constant $\Lambda>0$ depends on some norm of ${H}_S$ and what type of access to ${H}_S$ is given. The optimal (query) complexity of such Hamiltonian simulation algorithms is $\cO(t \Lambda + \log(1/\epsilon))$, where $\epsilon>0$ is the error. Hence, if the block-encoding can be constructed efficiently, the algorithm is efficient as desired. To this end, a commonly used access model for Hamiltonian simulation
is the sparse matrix model. In this model, the entries of a $d$-sparse matrix $A$ can be accessed efficiently as follows. There is a black-box that outputs the $(j,k)^{\rm th}$ entry of the matrix on input $(j,k)$. The black-box also outputs the locations of the $d$ nonzero entries in every row or column. Using this black-box a constant number of times, it is possible to build the block-encoding of $A/\Lambda$, where $\Lambda:=d \|A\|_{\max}$ in this case. This model can be adapted to the shadow Hamiltonian simulation approach. Basically, what we need is for ${H}_S$ to be $d$-sparse and efficiently accessible as explained. Three sufficient conditions that enable this are: i) that the number of operators in the linear combination $[H,O_m]=\sum_{m'} h_{mm'}O_{m'}$ is at most $d$ and $d=\polylog(M)$ or constant,  ii)
that each $h_{mm'}$ can be computed efficiently on input $(m,m')$, and iii) that the $m'$ corresponding to the $\ell^{\rm th}$ nonzero coefficient $h_{mm'}$ in the linear combination can also be computed efficiently, i.e., 
 $[H,O_m]=\sum_{\ell =1}^d h_{mm'(\ell)}O_{m'(\ell)}$.
 In this case, $\|{H}_S\|_{\max} = \max_{m,m'}|h_{m,m'}|$. Furthermore, the block-encoding 
 of ${H}_S/(d \|{H}_S\|_{\max})$ can be constructed efficiently under the conditions, giving an efficient Hamiltonian simulation algorithm by means of, e.g., quantum signal processing.

We explain how to meet the three conditions for the examples discussed in this article, but essentially the same techniques can be applied more broadly. Assume $S$ is the set of quadratic fermionic operators. The commutation between two such operators $O_m$ and $O_{m'}$ gives a linear combination of at most two other quadratic fermionic operators. Then, if  $H$ is a free-fermion Hamiltonian, we can assume that at most $p$ terms in $H$ do not commute with a given $O_m$, and hence $d=2p$. This occurs if the matrix of interaction strengths ${ \Gamma}$ is $d'$-sparse, so that $p \le 2 d'$: two quadratic fermionic operators $c_ic_j$ and $c_kc_l$ with non overlapping indices commute. 
Additionally, the entries of ${H}_S$   are proportional to the entries of ${ \Gamma}$, so that $\|{H}_S\|_{\max} \propto \|{ \Gamma}\|_{\max}$. Given sparse access to ${ \Gamma}$, it is also possible to compute and find the locations of all nonzero elements in the commutator $[H,O_m]$, by explicitly computing the $p$ nonzero commutations. This requires accessing ${ \Gamma}$, $\cO(p)$ times, resulting in an efficient Hamiltonian simulation algorithm of complexity $\cO(d' (t d' \|{ \Gamma}\|_{\max}+\log(1/\epsilon)))$. This is the number of queries to ${ \Gamma}$ and can be improved with a refined analysis.  

A similar analysis applies to the case of free-boson systems, showing that if the matrix of interaction strengths ${ \Gamma}$ is sparse and can be accessed efficiently, then 
${H}_S$ can be simulated efficiently. A more detailed study to do this simulation is given in Ref.~\cite{babbush2023exponential}, which assumes access to a mass matrix whose entries are the masses $m_j$, and a $d'$-sparse spring constant matrix whose entries are the spring constants $\kappa_{jk}$. The (query) complexity of simulating ${H}_S$ in this case was shown to be $\cO (t \sqrt{d'}\Lambda + \log(1/\epsilon) )$, where $\Lambda=\max_{i,j,k} \sqrt{\kappa_{jk}/m_i}$. 

The case of qubits, where $S$ is the set of all Pauli products, follows directly from the relation between the shadow state and $\ket \psi \otimes \overline{\ket \psi}$.
This implies ${H}_S= V_S (H \otimes \one_{2^n}+\one_{2^n} \otimes \bar H) (V_S)^\dagger$, where $V_S$ maps between basis states and Bell states and can be implemented with $\cO(n)$ two-qubit gates.
Then, ${H}_S$ can be efficiently simulated if $H$ can, by using standard Hamiltonian simulation with $H$ and $\bar H$. For $d$-sparse $H$, the complexity is $\cO (t d \|H\|_{\max} + \log(1/\epsilon) )$. 

Last, 
for the case of Lie algebras, the Hamiltonian is $H=\sum_k \alpha_k O_k$, with $\alpha_k \in \mathbb C$ and $O_k \in \mg$. We can satisfy the three conditions outlined above as long as $[O_j,O_k]=\sum_l f_{jkl}O_l$ is a linear combination of at most $d'$ operators, and the number of terms in $H$  that do not commute with any $O_j$ is at most $p$. If   we have access to a black-box that, given $j$, outputs all the $p$ terms in $H$ that do not commute with $O_j$ together with their indices, and another black-box that, given $(j,k)$, outputs all the $d'$ nonzero terms in 
$\sum_l f_{jkl}O_l$ together with their indices, we can use these black boxes $\cO(pd')$ times to create access to $d$-sparse ${H}_S$, where $d \le pd'$. This generalizes the case of quadratic fermionic operators in a natural way. The (query) complexity for simulating ${H}_S$ in the case of Lie algebras  is $\cO (t \Lambda + \log(1/\epsilon) )$, where $\Lambda=p d'\|{H}_S\|_{\max}$, and 
$\|{H}_S\|_{\max} \le \max_k |\alpha_k| \max_{jkl}|f_{jkl}|$.

\subsection*{Efficient preparation of some initial shadow states}
\label{sec:initialstates}

We also seek a procedure that prepares the initial
shadow state $\ket{\rho(0);S}$ efficiently.
This is not always possible, but for many interesting states that can be expressed succinctly, efficient quantum circuits might be constructed.

First, we can consider a model where a black-box unitary computes the expectations $\langle O_m \rangle$ on input $\ket m$. Starting from a superposition state like $\frac 1 {\sqrt M}\sum_{m=1}^M \alpha_m \ket m$, $\sum_m |\alpha_m|^2=1$, we can use the black-box and apply a conditional rotation to map
\begin{align}
    \frac 1 {\sqrt M}\sum_{m=1}^M \alpha_m \ket m \mapsto x \ket 0 \ket{\rho;S} + \ket{\rho;S^\perp} \;,
\end{align}
where $\ket{\rho;S^\perp}$ is orthogonal to $\ket 0 \ket{\rho;S}$ and $x\ge 0$ is the overlap. Standard techniques like amplitude amplification can be then used to distill $\ket{\rho;S}$ from this state with complexity $\cO(1/x)$.  The approach is efficient as long as $1/x=\cO(\polylog \; M)$.

Next we consider the case of fermions, where $S$ is the set of quadratic fermionic operators.
The free vacuum $\rho(0)=\rho= \ketbra{\rm vac}$ is the pure state containing no fermions.
In this state, the expectations of quadratic Majorana operators are
\begin{align}
\label{eq:barevacuum}
    \langle c_j c_k \rangle = \begin{cases}
        {\rm i} & (j,k) = (2l-1,2l) \ \text{for} \ l \in [n] ,\\
        0 & \text{otherwise}.
    \end{cases}
\end{align}
This follows from the relation between Majorana fermions and bare (free) fermions, such that $c_{2j-1}:=a^\dagger_j + a^{}_j$ and $c_{2j}:={\rm i}(a^\dagger_j - a^{}_j$) for all $j \in [n]$, where $a^\dagger_j$ ($a^{}_j$) are the fermionic creation (annihilation) operators of the  bare fermions. The vacuum state satisfies $a^{}_j \ket{\rm vac}=0$ for all $j$ and also $\bra{{\rm vac}} a^{}_ja^\dagger_j \ket{{\rm vac}}=\bra{{\rm vac}} (1-a^\dagger_ja^{}_j \ket{\rm vac}=1$.
The corresponding shadow state with respect to $S$ is (up to a global phase)
\begin{equation}
\label{eq:shadowvacuumfermion}
    \ket{\rho;S} = \frac{1}{\sqrt{n}} \sum_{l=1}^n |2l-1,2l\rangle \;.
\end{equation}
This state has a simple form, i.e., it is an equal superposition over basis states and can be prepared on a quantum computer in
time $\cO(\log(n))$ using elementary gates.
It is equivalent to the shadow state corresponding to the completely filled state, which is the state with one fermion in every mode.
This is implied by the fact that all operators in $S$ change sign under a particle-hole transformation, which places a fermion if the mode is empty or removes it if the mode is occupied~\cite{zirnbauer2021particle}: it
 maps $a^\dagger_j \mapsto a^{}_j$ and $a^{}_j \mapsto a^\dagger_j$. 
More generally, for this choice of $S=\{c_jc_k\}_{1\le j<k\le n}$, the shadow state of any $\rho$ is invariant under the particle-hole transformation (up to a global sign). This property can be useful to identify other shadow states that can also be prepared efficiently.
Indeed, other examples
of shadow states that can be prepared efficiently are those corresponding to fermionic product states $\ket{\psi_I} = \prod_{j \in I}a^\dagger_j \ket{\rm vac}$, $I \subseteq[n]$, that contains one fermion 
in the modes $I$ and no fermions elsewhere.
These   can be obtained from \cref{eq:shadowvacuumfermion}
by simply changing some signs of the amplitudes, which can be done efficiently if we have access to an oracle that specifies $I$ efficiently. The shadow states of
more general fermionic Gaussian states (Slater determinants) $\ket{\phi}$ could be prepared, for example, by means of Thouless's theorem~\cite{blaizot1986quantum}, which asserts
that $\ket \phi$ can be obtained as $e^{-\mathrm{i} \tilde H} \ket{\psi_I}$. 
Here, $\tilde H$ is also a free-fermion Hamiltonian.
If $\tilde H$
can be accessed efficiently as explained above, then the corresponding Hamiltonian $\tilde {H}_S$ required for mapping shadow states
can also be implemented efficiently. 
In general, there will be other shadow states that can also be prepared efficiently, and which do not necessarily correspond to fermionic Gaussian states.

The efficient preparation of certain shadow states for bosonic systems is already discussed in~\cref{sec:results}.
For qubit systems, if $\rho = \ketbra \psi$ is pure, shadow states with respect to the set $S$ of all Pauli operators are a simple transformation of $\ket \psi \otimes \overline{\ket \psi}$: $\ket{\rho;S}=V_S \ket \psi \otimes \overline{\ket \psi}$. These can be efficiently prepared as long as $\ket \psi$ (and $\overline{\ket \psi}$) can be efficiently prepared.

Last, for the case of Lie algebras, we generalize the case of fermions.  Assume further that  the Lie algebra $\mg$ is compact and semisimple. This enables us to present 
$\mg$ in the Cartan-Weyl decomposition. The Cartan subalgebra is a (maximal) set of commuting operators 
and the other terms are ``raising'' or ``lowering'' operators, playing a similar role to the fermionic creation and annihilation operators.
There is a special \emph{lowest-weight} state $\ket{{\rm lw}}$ --a generalization of the free vacuum state of fermions-- that has the following properties: it is an eigenstate of all diagonal operators in the Cartan subalgebra, and it is annihilated by all lowering operators.
If $h_i$, $i \in [r]$, are the diagonal operators, then $h_i \ket{{\rm lw}}= e_i \ket{{\rm lw}}$, and the shadow state of $\ket{{\rm lw}}$ is 
\begin{align}
  \ket{\rho;S}= \frac 1 {\sqrt A} \sum_{i=1}^r e_i \ket i \;.
\end{align}
The eigenvalues $e_i$ are purely imaginary and $A=\sum_i |e_i|^2$ is a normalization constant, which can often be computed efficiently. In principle, $r$ can be exponentially large. Nevertheless, the shadow state might be prepared efficiently if the $e_i$'s are described succinctly,
which turns out to be the case in many examples.
For example, in free fermions, the relavant algebra of quadratic fermionic operators is ${\mathfrak {so}}(2n)$. 
The lowest-weight state of ${\mathfrak {so}}(2n)$ is the free vacuum, where all eigenvalues of the diagonal operators are $-1$. As we discussed, the shadow state $\ket{\rho;S}$ in this case is a simple equal superposition state.

\bibliographystyle{unsrt}


\section*{Acknowledgements}

We thank Cristian Batista, Kipton Barros, David Gosset, Robert Huang, William Huggins,  Stephen Jordan, Marika Kieferova, Nicholas Rubin, and Nathan Wiebe for discussions.





\newpage
\onecolumn
\appendix
\section*{Supplementary information}

\label{sec:suppmat}

In this section we first provide the proofs of the main results, and later discuss broader properties of shadow states that are useful for their simulation and interpretation.

\subsection*{Supplementary Note 1: Main results}

{\bf Proof of Thm.~\ref{thm:main}.}
We begin by providing the proof to our main result in~Thm.~\ref{thm:main}. Let $U(t)$ be the evolution operator induced by $H$ and consider a general evolved state
$\rho(t) = U(t) \rho(0) U^\dagger(t)=\sum_\ell p_\ell \ketbra{\psi_\ell(t)}$,   where the $p_\ell$'s are probabilities and the $\ket{\psi_\ell(t)}$'s are pure states.
The amplitudes of $\ket{\rho(t);S}$ are proportional to 
$\langle O_m(t) \rangle=\sum_\ell p_\ell \bra{\psi_\ell(t)} O_m \ket{\psi_\ell(t)}$.
Schr\"odinger's equation $\frac{\rm d}{\rm dt}\rho(t)=-{\rm i}[H,\rho]$ and the IP imply 
\begin{align}
  \frac {\rm d} {{\rm d} t}  \langle O_m(t) \rangle& = 
  \sum_{\ell} p_{\ell}  \left(\frac {\rm d} {{\rm d} t}\bra{\psi_\ell(t)}\right)  O_m\ket{\psi_\ell(t)} +  \bra{\psi_\ell(t)}  O_m \left(\frac {\rm d} {{\rm d} t}\ket{\psi_\ell(t)}\right)\\
  &= \mathrm{i}  \sum_{\ell} p_{\ell} \bra{\psi_\ell(t)} [H,O_m]\ket{\psi_\ell(t)} \\
  &= -\mathrm{i} \sum_{m'=1}^M h_{mm'}  \langle O_{m'}(t) \rangle \;,
\end{align}
which is also \cref{eq:shadow_schrodinger}.
\qed

\vspace{.2cm}

{\bf Hermiticity of ${ H}_S$ in finite-dimensional systems.}
Next we show that the orthogonality of operators in $S$ ensures that the matrix ${ H}_S$ is Hermitian.
We provide the result in the following Lemma.

\begin{lemma}
    Let the system be finite-dimensional and the operators in $S$ be orthogonal, satisfying $\tr(O_m^\dagger O_{m'})=\lambda \delta_{mm'}$, for some $\lambda >0$. Then ${H}_S$ is Hermitian. 
\end{lemma}
\begin{proof}
   Under the orthogonality condition, the entries of
   ${H}_S$ satisfy
    \begin{equation}
        h_{mm'} = -\frac 1 \lambda \tr(O^\dagger_{m'} [H, O_m]) = \underbrace{-\frac 1 \lambda \overline{\tr([O^\dagger_m, H] O_{m'})}}_{\text{(take Hermitian conjugate)}} = \underbrace{- \frac 1 \lambda\overline{\tr(O^\dagger_m [H, O_{m'}])}}_{\text{(cyclic property of trace)}} = \overline{h_{m'm}}.\vspace{-2em}
    \end{equation}
\end{proof}

\vspace{.2cm}

{\bf Proof of Thm.~\ref{thm:main2}.}
We now discuss the main result that shows how to encode two-time correlators and more complex expectations in quantum states.
    Let ${\bf O}(t) =(O_1(t),\ldots,O_M(t))^T$ and consider the vector of $M^2$ operators ${\bf O}(t) \otimes {\bf O}(t')$. Heisenberg's equations imply
    \begin{align}
        \frac{\partial}{\partial t} {\bf O}(t) \otimes {\bf O}(t') =\mathrm{i} [H, {\bf O}(t)] \otimes {\bf O}(t') = -\mathrm{i} \big({ H}_S {\bf O}(t)\big) \otimes {\bf O}(t') = -\mathrm{i} \big({ H}_S \otimes \one_M) \big({\bf O}(t) \otimes {\bf O}(t')\big) \;.
    \end{align}
    A corresponding equation can be obtained applying $\frac{\partial}{\partial t'}$.
    We can compute the expectation with respect to the initial state $\rho$ on either side 
    of this equation and obtain~\cref{eq:shadowproduct_schrodinger}, since we can alternatively express $\ket{\rho;S(t,t')}$ as the vector $\frac 1 {\sqrt A} \langle {\bf O}(t) \otimes {\bf O}(t') \rangle=\frac 1 {\sqrt A} \sum_{m,m'}\langle O_m(t) O_{m'}(t') \rangle \ket{m,m'}$.
\qed

\vspace{.2cm}

{\bf Shadow Hamiltonian simulation for quantum circuits.}
As discussed in~\cref{sec:formalism}, the shadow Hamiltonian simulation approach can be extended to the case of quantum circuits in a simple way.
Let $G=G^L\ldots G^1$ be a quantum circuit described as a sequence of gates $G^j$, $j \in [L]$. The analogous property to the invariance property discussed in~\cref{sec:formalism} in this case is 
\begin{equation}\label{eq:invariance_condition_unitary}
    (G^j)^\dagger O_m G^j=\sum_{m'} g^j_{mm'} O_{m'},
\end{equation}
for all $O_m \in S$, $m \in [M]$,  and $j \in [L]$, where $g^j_{mm'} \in \mathbb C$. These coefficients define an $M \times M$ matrix ${ G}^j_{S}$ for each $G^j$. If these matrices are also unitary---the analogous condition to requiring ${ H}_S$ be Hermitian---then we can simulate them on a quantum computer and perform maps between the corresponding shadow states.

\vspace{.2cm}

{\bf Time-dependent Hamiltonians and quantum circuits.}
In~\cref{sec:formalism} we described a way to encode time-evolving operators in quantum states. We provided details of the result for time-independent Hamiltonians and here we show the result for the time-dependent case. To this end, it suffices to consider the case of quantum circuits, since gates can be thought of as time evolutions with different Hamiltonians.

Let $G=G^L \ldots G^1$ be the quantum circuit composed of $L$ unitary gates $G^j$. The requirement for these gates and $S$ to satisfy the IP is given in \cref{eq:invariance_condition_unitary}. Accordingly, to study the evolution of an operator $Z=\sum_{m=1}^M z_m O_m$ as we act with the gates, we define the operators $Z_l:=(G^{L+1-l})^\dagger \ldots (G^L)^\dagger Z G^L \ldots G^{L+1-l}$ for all $l \in [L]$. If the IP is satisfied, then $Z_l =\sum_{m=1}^M z_{l,m} O_m$, where $z_{l,m} \in \mathbb C$. Also $Z_0:=Z$, i.e., $z_{0,m}=z_m$ for all $m \in[M]$. 
These $Z_l$'s are basically a (backwards) transformation of
$Z$ when acting with the last $l$ gates of the circuit, and we are interested in $Z_L$, which is the operator transformed by all gates.

Then, \cref{eq:invariance_condition_unitary} gives
the transformation on the coefficients:
\begin{equation}
\label{eq:schrodingercoeff2}
   Z_{l+1}=(G^{L-l})^\dagger Z_l G^{L-l}=\sum_{m=1}^M z_{l,m} \sum_{m'=1}^M g^{L-l}_{mm'}O_{m'} \implies  z_{l+1,m}=\sum_{m'=1}^M z_{l,m'} g_{m'm}^{L-l} \;.
\end{equation}
Let ${ G}^j_{S}$ be the matrix of entries $g_{mm'}^j$.
When this is unitary---the analogue of requiring ${ H}_S$ be Hermitian---\cref{eq:schrodingercoeff2} denotes a unitary transformation on the vector of coefficients $z_{l,m}$, for each $l$. We then consider the operator-vector mapping $Z_l \mapsto \ket{Z_l} \propto \sum_{m=1}^M z_{l,m} \ket m$. This provides an encoding of the time-evolved operators $Z_l$ in a quantum state, and \cref{eq:schrodingercoeff2} implies
\begin{align}
    \ket{Z_{l+1}} =({ G}^{L-l}_{S})^T \ket{Z_l} \;, \quad l=0,\ldots,L-1 \;.
\end{align}
A quantum algorithm that prepares $\ket{Z_L}$
is then a sequence of the unitaries $({ G}^{1}_{S})^T \ldots ({ G}^{L}_{S})^T$ acting on $\ket{Z_0}$. 
Note that applying the transpose of ${ G}^{L-l}_{S}$ is the analogue
to evolving with $\overline{{ H}_S}(t-s)$, when we consider the case of time-dependent Hamiltonians in continuous time. This is the result stated in~\cref{sec:formalism}.

\subsection*{Supplementary Note 2: Shadow Hamiltonian simulation for free-fermion systems} \label{sec:freefermiondetails}

The purpose of this section is to analyze
the details of each step in shadow Hamilton simulation on an example, that of free-fermion systems. In the main text we focused on the expectations of products of two fermionic operators for simplicity, but we start the analysis by considering single operators. However, the approach generalizes to expectations of higher-order operators in a straightforward way, by following Thm.~\ref{thm:main2}, as we explain. 

Like any quantum algorithm for quantum simulation, shadow Hamiltonian simulation consists of three steps: i) initial shadow state preparation, ii) evolution, and iii) measurement. The efficient preparation of initial shadow states has already been discussed in~\cref{sec:methods}. The measurement of an observable on the final shadow state in this example can approximate a function of the expectations. However, it is perhaps more interesting to obtain the overlap between the final shadow state and some target state, as this will be a linear combination of the expectations. Such linear combinations can describe properties like energies of the system at later times, as discussed in~\cref{sec:formalism}. Standard techniques can be used to perform these measurements. 
Hence, the key part we analyze in this section
is the efficient simulation of the evolution, and the assumptions on the Hamiltonians to this end.

\subsubsection*{Free-fermion Hamiltonians and the invariance property}

The Hamiltonian of a free-fermion system is often presented as 
\begin{align}
\label{eq:fermionicHamiltonian}
    H = \sum_{j,k=1}^n \alpha_{jk} a^\dagger_j a^{\!}_k + \beta_{jk} a^{\!}_j a^{\!}_k -\overline{\beta_{jk}} a^\dagger_j a^\dagger_k \;,
\end{align}
where $a^\dagger_j$ ($a^{\!}_j$) are fermionic creation (annihilation) operators satisfying anticommutation relations
$\{a_j,a_k\}=\{a^\dagger_j,a^\dagger_k\}=0$
and $\{a_j,a^\dagger_k\}=\delta_{jk}$, where $\{a,b\}=ab+ba$. The coefficients 
$\alpha_{jk},\beta_{jk}\in\mathbb{C}$ are ``interaction strengths'' and satisfy $\alpha_{kj}=\overline{\alpha_{jk}}$. To ease the exposition, we equivalently assume a presentation of $H$ in terms of $2n$ Majorana operators, which are Hermitian and defined by
\begin{equation}
    c_{2j-1} := a^\dagger_j + a^{\!}_j \;, \quad c_{2j} :=\mathrm{i} (a^\dagger_j - a^{\!}_j) \;, \quad j = 1,\dots,n \;.
\end{equation}
These satisfy the anticommutation relations
\begin{equation}
\label{eq:fermionanticomm}
    \{c_j,c_k\} = 2  \delta_{jk} \;, \quad j,k \in \{1,\dots,2n\} \;.
\end{equation}
Taking $H$ to this form requires minimal preprocessing.
Hence,  \cref{eq:fermionicHamiltonian} becomes
\begin{equation} \label{eq:majorana_coeffs}
    H = \sum_{j,k=1}^{2n} \gamma_{jk} c_j c_k \;,
\end{equation}
where the interaction strengths $\gamma_{jk}\in \mathbb{C}$
define a $2n \times 2n$ Hermitian matrix ${ \Gamma}$. This is the same Hamiltonian discussed in~\cref{sec:formalism}.

Majorana operators satisfy the IP:
\begin{align}
\label{eq:fermioniccommutation}
  [c_jc_k,c_l]=2\delta_{lk} c_j - 2\delta_{lj} c_k,
\end{align}
where $[a,b]:=ab-ba$ are the commutators. These follow directly from \cref{eq:fermionanticomm}.
Furthermore, Majorana operators can also be shown to be orthogonal in $\mathbb C^{2^n}$, the Hilbert space that $H$ acts on:
\begin{align}
    \tr(c_jc_k)= 2^n \delta_{jk} \;.
\end{align}
This also follows directly from \cref{eq:fermionanticomm} and the cyclic property of the trace. (It is also well known that Majorana operators can be mapped to products of Pauli operators acting on $n$ qubits using the Jordan-Wigner transform~\cite{jordan1928paulische}, and these Pauli operators are orthogonal.)

\subsubsection*{The Hamiltonian  {${{H}}_S$}}

In terms of Majorana operators, $H$ and the set $S=\{c_j\}_j$ satisfy the IP. In addition, these operators are orthogonal, implying that ${ H}_S$ is Hermitian and of dimension $M \times M=2n \times 2n$. Note that, in the theory of Lie algebras, ${ H}_S$ is the ``adjoint representation'' of $H$.

We proceed to obtain the matrix entries of
${ H}_S$. In particular, the $l^{\rm th}$ row of ${ H}_S$ can be determined from
\begin{align}
    [H,c_l]=\sum_{j,k=1}^{2n} \gamma_{jk} [c_jc_k,c_l] = \sum_{j,k=1}^{2n} \gamma_{jk}(2 \delta_{kl}c_j - 2 \delta_{jl}c_k) = \sum_{j=1}^{2n} 2 \gamma_{jl} c_j - \sum_{k=1}^{2n} 2 \gamma_{lk} c_k =\sum_{k=1}^{2n} 2   (\gamma_{kl}-\gamma_{lk}) c_k\;.
\end{align}
Let ${\bf c}:=(c_1,\ldots,c_{2n})^T$. 
Then, in compact form we obtained
\begin{align}
    [H,{\bf c}] =  - 2 ({ \Gamma} -{ \Gamma}^T)  {\bf c}= - 2 ({ \Gamma} -{\bar{ \Gamma}})  {\bf c} \;.
\end{align}
It follows that, according to the definition
in~\cref{eq:invarianceproperty},
\begin{align}
\label{eq:freefermionHS}
    { H}_S = 2 ({ \Gamma} -{\bar{ \Gamma}}) \;,
\end{align}
which is also Hermitian. 
The $(j,k)^{\rm th}$ entry of ${ H}_S$
is then $4 {\rm i} \im (\gamma_{jk})$.

\subsubsection*{Simulation of {${\emph{H}}_S$}}

Two direct consequences of \cref{eq:freefermionHS}
are: i) if ${ \Gamma}$ is $d'$-sparse, 
then so is the matrix ${ H}_S$, and ii) the largest entry of ${ H}_S$ in magnitude satisfies $\|{ H}_S\|_{\max} \le 4 \|{ \Gamma}\|_{\max}$. Furthermore, 
we can obtain efficient sparse access to ${ H}_S$ if sparse access to ${ \Gamma}$ is given.

Hence, to simulate ${ H}_S$ for time $t$ and precision $\epsilon$ in this case, we can use a method like quantum signal processing (QSP)~\cite{LC19}. Note that the dimension of ${ H}_S$ is $2n \times 2n =M \times M$, i.e., $M=2n$. The query complexity of QSP 
is
\begin{align}
    \cO \left(d' t \|{ H}_S\|_{\max} + \log(1/\epsilon) \right) =  \cO \left(d' t \|{ \Gamma}\|_{\max} + \log(1/\epsilon) \right) 
\end{align}
and the gate complexity is a multiplicative factor of $\cO(\log(n d' t \|{ H}_S\|_{\max}/\epsilon))=\cO(\log(n d' t \|{ \Gamma}\|_{\max}/\epsilon))$ of the query complexity.
Thus, shadow Hamiltonian simulation is efficient under broad assumptions on $H$ or $ \Gamma$.

QSP simulates the Hamiltonian by first constructing a block-encoding. The details can be found in prior work (cf.,~\cite{LC19}). Here, 
we sketch a standard technique for block-encodings for completeness that implies the stated complexities. In the sparse matrix model, we assume we are given an oracle $\cS_{ \Gamma}$ such that it performs the following two transformations:
\begin{align}
  \cS_{ \Gamma} \ket{j,k,z} &\mapsto   \ket{j,k,z \oplus \gamma_{jk}} \; , \; (j,k) \in \{1,\ldots,2n\} \;,\\
   \cS_{ \Gamma} \ket{j,\ell}& \mapsto   \ket{j,\nu(j,\ell)} \; , \; j \in \{1,\ldots,2n\} \;, \ell \in \{1,\ldots,d'\}\;.
\end{align}
Here, $z$ and $\gamma_{jk}$ are two bit strings of some fixed length ${\rm k}$. For our purposes, it suffices to work within precision so that $2^{-{\rm k}}=\cO(\epsilon/(td'\|{ H}_S\|_{\max}))=\cO(\epsilon/(td'\|{ \Gamma}\|_{\max}))$.
Also $\gamma_{jk}$ inside the ket notation is a bit string that
approximates the $(j,k)^{\rm th}$ entry of ${ \Gamma}$. Then, the first map produces the individual entries of the matrix. Also, $\ell$ denotes the location of $\ell^{\rm th}$ nonzero entry in the $j^{\rm th}$ row of ${\Gamma}$.

It is straightforward to construct an oracle $\cS_{{ H}_S}$ for ${ H}_S$ in this case. Essentially, we only need to modify the first map so that
\begin{align}
  \cS_{{ H}_S} \ket{j,k,z} &\mapsto   \ket{j,k,z \oplus h_{jk}} \; , \; (j,k) \in \{1,\ldots,2n\} \;,
\end{align}
where $h_{jk}$ is a bit string of size ${\rm k}$ that approximates the corresponding entry $4 {\rm i} \im (\gamma_{jk})$.
This can be done with simple manipulations of the bits that represent $\gamma_{jk}$, and the complexity of this step is subdominant.

Equipped with an oracle for ${ H}_S$,
there are known techniques for constructing the block-encoding. Define the following unitary transformation acting on $\mathbb C^{2n \times 4n}$:
\begin{align}
    T \ket j \ket 0 \mapsto \ket j \frac 1 {\sqrt{d'}}\sum_{k \in \{1,\ldots,2n\}: h_{jk}\ne 0}  \sqrt{\frac{\overline{ h_{jk}}}{\|{ H}_S\|_{\max}}} \ket k + \sqrt{1- \frac{ |\overline{h_{jk}}|}{\|{ H}_S\|_{\max}}} \ket{k+2n }\; , \; j \in \{1,\ldots,2n\} \;.
\end{align}
Let $S$ be the swap operation that permutes the two registers. Then, it can be shown
\begin{align}
   T^\dagger S T = \begin{pmatrix}
       \frac 1{d'\|{ H}_S\|_{\max}} { H}_S  & . \cr . & .
   \end{pmatrix} 
\end{align}
in the space $\mathbb C^{2n \times 4n}$. That is,
$ T^\dagger S T$ is a block-encoding of $ \frac 1{d'\|{ H}_S\|_{\max}} { H}_S$, which is the operation that QSP uses. Importantly,
$T$ (and $T^\dagger$) can be implemented with a constant number of uses of $\cS_{{ H}_S}$ or $\cS_{{ \Gamma}}$,
and a number of two-qubit gates that is $\cO(\log(n))$ to implement the swap gates and other $\cO({\rm k})=\cO(\log(td'\|{ H}_S\|_{\max}/\epsilon))$ two-qubit gates for implementing $T$ from, for example, inequality tests for state preparation using oracles~\cite{SandersPRL18}.

Note that in specific applications, 
having access to the oracle essentially means
being able to efficiently compute the interactions
between modes, for every given mode $j$. That is, we need an efficient description of the physical system and its interactions that avoids incurring in exponential complexity. This occurs in many instances, for example those where the interaction strengths are uniform or are obtained from a limited set of possible values. In more general applications, however, if all interaction strengths have to 
be specified, constructing and implementing the oracle
would be inefficient.

\subsubsection*{Expectations of quadratic operators and beyond}

Thus far we discussed the case where $S=\{c_j\}_j$ is the set of Majorana operators.
As explained in Thm.~\ref{thm:main2}, we can leverage the result and consider expectations of higher-order operators, like quadratic Majorana operators where $S=\{c_jc_k\}_{j,k}$. Since we can let $j,k$ range over $\{1,\ldots,2n\}$, some of these quadratic operators might be linearly dependent. For example, $(c_j)^2 = \one_{2^n}$, the identity operator in $\mathbb C^{2^n}$. 

Nevertheless, according to Thm.~\ref{thm:main2}
we can consider the shadow state
\begin{align}
    \ket{\rho;S}=\frac 1 {\sqrt A}\sum_{j,k=1}^{2n}\langle c_j c_k \rangle \ket {j,k}
\end{align}
as an example. Evolving this in time would produce
$\ket{\rho(t);S}=\frac 1 {\sqrt A}\sum_{j,k}\langle c_j c_k (t)\rangle \ket {j,k}$, which encodes the expectations of quadratic operators at a later time. The evolution of this shadow state satisfies
\begin{align}
   \frac{\rm d}{\rm dt} \ket{\rho(t);S} = -{\rm i} ({ H}_S \otimes \one_{2n} + \one_{2n} \otimes { H}_S)\ket{\rho(t);S} \;,
\end{align}
where ${ H}_S$ is the Hamiltonian described in \cref{eq:freefermionHS}. Since this 
applies ${ H}_S$ individually to each system, the complexity of this simulation is similar to the one analyzed. Similarly, we can go beyond quadratic operators by means of Thm.~\ref{thm:main2}.

\subsection*{Supplementary Note 3: Additional properties of the shadow state} \label{sec:properties}

We now discuss some properties and alternate representations of the shadow state $\ket{\rho;S}$ for finite dimensional systems, without reference to any time-evolution problem one might wish to solve with this state.

\paragraph{Standard basis.} 
Let $\rho = \ketbra{\psi} \in \C^{N \times N}$ be a pure quantum state of a system of dimension $N <\infty$, and let $S$ be an orthogonal basis of the space $\C^{N \times N}$. In this case, $\ket{\rho;S}$ is related to other standard representations of the state $\ket{\psi}$. One example is the standard basis,  $S_\mathrm{SB}=\{O_{ij}:i,j=0,\ldots,N-1\}$, where $O_{ij} = \ket{j}\bra{i}$. Then,
\begin{equation}\label{eq:sb}
    \ket{\rho;S_\mathrm{SB}} 
    = \sum_{i,j=0}^{N-1} \langle O_{ij} \rangle \ket{i,j}
    = \sum_{i,j=0}^{N-1} \bigl( \braket{\psi}{j} \braket{i}{\psi} \bigr) \ket{i,j}
    = \sum_{i=0}^{N-1} \ket{i} \braket{i}{\psi} \otimes \sum_{j=0}^{N-1} \ket{j} \overline{\braket{j}{\psi}}
    = \ket{\psi} \otimes \ket{\barpsi}.
\end{equation}

This is already an interesting state that is different from the ``standard'' representation $\ket{\psi}$. Clearly we can create $\ket{\psi}$ given one copy of $ \ket{\psi} \otimes \ket{\barpsi}$, but the reverse is not possible. In fact, there exist problems that can be solved using only one copy of $\ket{\psi}\otimes \ket{\barpsi}$ that would require exponentially many copies of $\ket{\psi}$ alone to solve.

\begin{theorem}\label{thm:separation}
    Consider the problem of deciding if an $n$-qubit state $\ket{\psi}$ satisfies $|\braket{\psi}{\barpsi}|\leq 1/3$ or $|\braket{\psi}{\barpsi}|\geq 2/3$, promised that one of these holds. Any quantum algorithm that solves this problem using only copies of $\ket{\psi}$ must use $\Omega(2^{n/2})$ copies, but the problem can be solved with one copy of $\ket{\psi} \otimes \ket{\barpsi}$ with constant success probability.
\end{theorem}
\begin{proof}
    The problem can be solved with constant success probability using one copy of $\ket{\psi} \otimes \ket{\barpsi}$ by using the well-known swap test. The success probability can be made arbitrarily close to $1$ using a constant number of copies of $\ket{\psi} \otimes \ket{\barpsi}$ by standard error reduction via majority vote.

    For the lower bound, we observe that any state with all real amplitudes satisfies $|\braket{\psi}{\barpsi}|=1$ and that a Haar-random state has $|\braket{\psi}{\barpsi}|=o(1)$ with high probability. The lower bound follows from Ref.~\cite{BS19} that shows that a binary phase state (i.e., a uniform superposition state where all phases are uniformly random in $\{\pm 1\}$) is indistinguishable from a Haar-random state unless we have $\Omega(2^{n/2})$ copies of the state. Similarly, \cite{CHM21} shows that distinguishing a Haar-random state with real entries from a standard Haar-random state requires $\Omega(2^{n/2})$ copies of the state.
\end{proof}

The state $\ket{\psi}\otimes \ket{\barpsi}$ is also related to another well-known alternate representation of quantum states that uses only real amplitudes. In this ``real representation'' of quantum states, the state $\ket{\psi}$ is represented by the quantum state $\ket{0}\ket{\Re(\psi)} + \ket{1}\ket{\Im(\psi)}$, where $\ket{\Re(\psi)}$ and $\ket{\Im(\psi)}$ are the quantum states where each amplitude is replaced by the real and imaginary parts of the corresponding amplitude of $\ket{\psi}$, respectively. Up to a Hadamard on the first qubit, this state is also $\ket{0}\ket{\psi}+\ket{1}\ket{\barpsi}$.

\paragraph{Orthonormal bases.}
For the standard basis $S_\mathrm{SB}$, the shadow state is $\ket{\psi} \otimes \ket{\barpsi}$.
More generally, the shadow state for any orthonormal basis $S$ is a unitary transformation $V_S$ applied to $\ket{\psi} \otimes \ket{\barpsi}$.

\begin{lemma}\label{lem:on}
    Let $S = \{O_{ij}: i,j=0,\ldots,N-1 \}$ be an orthonormal basis of $\C^{N \times N}$. Let $V_S$ be the unitary on $\C^N \otimes \C^N$ defined by 
    \begin{equation}\label{eq:Vtransformation}
    V^\dagger_S \ket{i,j} = (\id_N \otimes O_{ij}) \sum_{k =0}^{N-1} \ket{k,k}    .
    \end{equation}    
    Then, if $\rho = \ketbra{\psi}$ is pure, we have
    \begin{equation}
        \ket{\rho;S} = V_S \ket{\psi}\otimes \ket{\barpsi}.
    \end{equation}
\end{lemma}
\begin{proof}
    First, we show that $V_S$ is indeed unitary. Observe that
    \begin{equation}
        \bra{i',j'} V^{}_S V_S^\dagger \ket{i,j} 
        = \sum_{k,k'=0}^{N-1} \bra{k',k'} (\id_N \otimes O^\dagger_{i'j'}O^{}_{ij}) \ket{k,k} 
        = \sum_{k=0}^{N-1} \bra{k} O^\dagger_{i'j'} O^{}_{ij} \ket{k} 
        = \tr(O^\dagger_{i'j'} O^{}_{ij})
        = \delta_{ii'}\delta_{jj'},
    \end{equation}
    since the operators $O_{ij}$ are orthogonal. We obtain
    \begin{equation}
        \langle O_{ij} \rangle 
        =  \tr( \ket{\psi}\bra{\psi} O_{ij}) 
        =  \sum_{k=0}^{N-1} \braket{k}{\psi}\bra{\psi} O_{ij} \ket{k}  
        = \sum_{k=0}^{N-1} \braket{k}{\psi} \bra{k} O^\dagger_{ij} \ket{\barpsi}
        = \sum_{k=0}^{N-1} \bra{k,k} (\id_N \otimes O^\dagger_{ij}) \ket{\psi}\otimes \ket{\psibar}\;.
    \end{equation}
    Using the definition of $V_S$, we get that $\langle O_{ij} \rangle = \bra{i,j}V_S\ket{\psi}\otimes\ket{\barpsi}$, which gives
    \begin{equation}
        \ket{\rho;S} = \sum_{i,j=0}^{N-1} \langle O_{ij} \rangle \ket{i,j} 
        = \sum_{i,j=0}^{N-1} \bra{i,j}V_S\ket{\psi}\otimes\ket{\barpsi} \ket{i,j} 
        = \Bigl(\sum_{i,j=0}^{N-1} \ket{i,j} \bra{i,j}\Bigr) V_S\ket{\psi}\otimes\ket{\barpsi} 
        = V_S\ket{\psi}\otimes\ket{\barpsi}. \qedhere
    \end{equation}
\end{proof}

So far we have only considered pure states. But observe that the lemma can equivalently be restated as asserting that $\ket{\rho;S} = V_S \ket{\rho;S_\mathrm{SB}}$. Since both sides of the equation are linear in $\rho$, we have the straightforward corollary that $\ket{\rho;S} = V_S \ket{\rho;S_\mathrm{SB}}$ holds for all mixed states $\rho$.

To understand the form of $\ket{\rho;S_\mathrm{SB}}$, note that \cref{eq:sb} equivalently implies that $\ket{\rho;S_\mathrm{SB}}$ equals $\mathrm{vec}(\rho)$ when $\rho = \ketbra{\psi}$, where $\mathrm{vec}$ is the operation that maps $\sum_{j,k} a_{jk} \ket{j}\bra{k} \mapsto \sum_{j,k} a_{jk} \ket{j}\ket{k}$ ~\cite{watrous2018theory}. When $\rho$ is mixed, $\ket{\rho;S_\mathrm{SB}}$ is the state whose entries are proportional to $\mathrm{vec}(\rho)$ with the normalization constant equal to the square root of the purity of $\rho$, $\sqrt{\tr(\rho^2)}$. Alternatively, if $\rho = \sum_\ell p_\ell \ketbra{\psi_\ell}$, then $\ket{\rho;S_\mathrm{SB}} \propto \sum_\ell p_\ell \ket{\psi_\ell}\otimes \ket{\overline{\psi_\ell}}$.

\paragraph{Examples of $V_S$.} 

When $S=S_\mathrm{SB}$, observe that $V_S\equiv \id_{N^2}$, where $N=2^n$. More interestingly, let $S$ be the set of Pauli operators on $n$ qubits, $\{P_{ij}\}_{i,j \in \{0,1\}^n}$, 
where $P_{ij}$ are tensor products of $n$ Pauli operators as discussed in~\cref{sec:formalism}. In this case, the state $(\id_N \otimes P_{ij})\frac 1 {\sqrt N} \sum_{k=0}^{N-1}  \ket{k,k}$ are tensor products of $n$ Bell pairs. Thus the operator $V_S$ is a Bell rotation and can be achieved by the sequence of $n$ entangling gates shown in~Fig. 1.

As another example, let $S$ be the set of Heisenberg-Weyl operators of dimension $N$. This is a group of $N^2$ unitary matrices spanned by the generators $X$ and $Z$, where $X\ket{j} = \ket{j+1 \pmod N}$ and $Z\ket{j}=e^{-\mathrm{i} 2\pi j/N}\ket{j}$. In this case, the unitary $V_S$ can be constructed from the discrete Fourier transform matrix of size $N \times N$, and the conditional shift operator $\sum_{j=0}^{N-1} \ketbra{j} \otimes X^j$, both of which have efficient quantum circuits.
In these examples, the unitary $V_S$ was efficient to implement as a quantum circuit. In general, this may not be the case.

\paragraph{Incomplete bases.}
In many of the applications considered in this article, $S$ is not a complete basis of the space, but an incomplete set of orthonormal operators. In this case, $\ket{\rho;S}$ may not contain all the information contained in $\ket{\psi}$ and hence it will not simply be a unitary transformation applied to $\ket{\psi}\otimes \ket{\barpsi}$. However, the state $\ket{\rho;S}$ can still be viewed as a unitary transformation and projection of this state, justifying the name ``shadow'' state. 

More precisely, even when $S$ is not a complete basis, the same proof of \Cref{lem:on} can be followed, except for the last equation where $\sum_{i,j=0}^{N-1} \ket{i,j} \bra{i,j}$ is not identity in this case: the sum is not over all $i,j$. This sum becomes an orthogonal projector $P_S$. Thus, when $\rho=\ketbra{\psi}$ is pure, we have $\ket{\rho;S} = P_S V_S \ket{\psi}\otimes \ket{\barpsi}$. 
In a way, the shadow state can be interpreted as the ``shadow'' of $\ket{\psi}\otimes \ket{\barpsi}$ when projected onto a particular subspace specified by the operators in $S$.

\paragraph{Using shadow states to compute overlaps.}
Shadow states of mixed states can be used to obtain overlaps. Since these can be seen as ``purifications'', it is possible to use quantum metrology techniques to improve the complexity dependence of overlap estimation in terms of precision~\cite{KOS07}. More precisely, we show the following result.

\begin{lemma}
    \label{lem:overlap}
    Let $S=\{O_{ij}:i,j =0,\ldots,N-1\}$ be an orthonormal basis of 
    $\mathbb C^{N\times N}$.
    Let $\rho=\sum_\ell p_\ell \ketbra {\psi_\ell}$ and $\sigma=\sum_\ell q_\ell \ketbra{\phi_{\ell}}$ be two mixed states in $\mathbb C^N$, with $\sum_\ell p_\ell =\sum_\ell q_\ell=1$, $p_\ell>0$, and $q_\ell>0$. 
    Consider the subnormalized states $\ket{\tilde \rho;S}:=V_S \sum_\ell p_\ell \ket {\psi_\ell} \otimes  \ket {\overline{ \psi_\ell}}$ and $\ket{\tilde \sigma;S}:=V_S \sum_\ell q_\ell \ket {\phi_\ell} \otimes  \ket {\overline{ \phi_\ell}}$. The amplitudes of these states coincide with the expectations of $O_{ij}$ on $\rho$ and $\sigma$, respectively. 
    Then, the average overlap satisfies
    \begin{align}
        \tr(\rho \sigma)=\bra{\tilde \rho_\cV}{\tilde \sigma_\cV}\rangle   \;.
    \end{align}
\end{lemma}
\begin{proof}
    Since $V_S$ is unitary, explicit calculation gives
    \begin{align}
         \bra{\tilde \rho_\cV}{\tilde \sigma_\cV}\rangle & = \sum_{\ell \ell'}p_\ell q_{\ell'} (\bra {\psi_\ell} \otimes  \bra {\overline {\psi_\ell}} )\; (\ket {\phi_{\ell'}} \otimes  \ket {\overline{\phi_{\ell'}}}) \\
         & = \sum_{\ell \ell'}p_\ell q_{\ell'}  |\bra {\psi_\ell}{\phi_{\ell'}}\rangle|^2 \\
         & = \tr(\rho \sigma)\;.
    \end{align}
\end{proof}

Hence, it is possible to obtain the average overlap using standard quantum metrology techniques when having access to the operations that prepare $\ket{\tilde \rho;S}$ and $\ket{\tilde \sigma;S}$.
The result can be generalized further to account for unitary transformations.

\begin{lemma}
    \label{lem:overlap2}
    Let $S=\{O_{ij}:i,j =0,\ldots,N-1\}$ be an orthonormal basis of 
    $\mathbb C^{N\times N}$.
    Let $\rho=\sum_\ell p_\ell \ketbra {\psi_\ell}$ and $\sigma=\sum_\ell q_\ell \ketbra{\phi_\ell}$ be two mixed states in $\mathbb C^N$, with $\sum_\ell p_\ell =\sum_\ell q_\ell=1$, $p_\ell>0$, and $q_\ell>0$. 
    Consider the subnormalized states $\ket{\tilde \rho;S}:=V_S \sum_\ell p_\ell \ket {\psi_\ell} \otimes  \ket {\overline {\psi_\ell}}$ and $\ket{\tilde \sigma_\cV}:=V_S \sum_\ell q_\ell \ket {\phi_\ell} \otimes  \ket {\overline {\phi_\ell}}$. The amplitudes of these states coincide with the expectations of $O_{ij}$ on $\rho$ and $\sigma$, respectively. Then, if $U$ is a unitary acting on $\mathbb C^N$,
    \begin{align}
    \tr(\rho U \sigma U^\dagger) & = 
        \bra{\tilde \rho_\cV}{\tilde U}\ket{\tilde \sigma_\cV} \;,
    \end{align}
    where ${\tilde U}:=V_S (U \otimes \overline{U}) (V_S)^\dagger$ is a unitary acting on ${\mathbb C^{N^2}}$.  In particular, if $\rho=\sigma=\ketbra \psi$ are pure states,
    \begin{align}
     |\bra \psi U \ket \psi|^2&= \bra{\rho;S}{\tilde U}\ket{\rho;S}   \;,
    \end{align}
    where the shadow state is 
    $\ket{\rho;S}=\ket{\tilde \rho;S}$ in this case, since $\ket{\tilde \rho;S}$ has unit norm.
\end{lemma}
\begin{proof}
    Since $U \sigma U^\dagger$ is a mixed state, from  Lemma~\ref{lem:overlap} we have
    \begin{align}
    \label{eq:traceinner}
       \tr(\rho U \sigma U^\dagger) & =    
        \bra{\tilde \rho;S}  (U \tilde \sigma U^\dagger);S \rangle \;,
    \end{align}
    where the amplitudes of the (subnormalized) state $\sket{(U\tilde \sigma U^\dagger);S}$ coincide with the expectations of $O_{ij}$ in 
     $U \sigma U^\dagger$. 
    By definition, ${\tilde U}:=V_S ( U \otimes \overline{U})(V_S)^\dagger$,
    where $V_S$ is a unitary that has been defined in~\cref{eq:Vtransformation}. 
    Also, we proved for orthonormal bases, 
    \begin{align}
        \sket {(U\tilde \sigma U^\dagger);S}=\sket {\left(\sum_{\ell} q_{\ell}U \ketbra {\phi_\ell} U^\dagger\right);S} = V_S \sum_\ell q_\ell (U \ket {\phi_\ell}) \otimes (\overline{ U \ket { \phi_\ell}}) \;.
    \end{align}
    Then,
    \begin{align}
        \sket {(U\tilde \sigma U^\dagger);S}& = V_S (U \otimes \overline{U}) \sum_\ell q_\ell \ket {\phi_\ell}  \otimes  \ket {\overline{ \phi_\ell}} \\
        & = V_S (U \otimes \overline{U}) (V_S)^\dagger V_S \sum_\ell q_\ell \ket {\phi_\ell}  \otimes  \ket {\overline{ \phi_\ell}}\\
        & = {\tilde U} \sket {\tilde \sigma;S}\;.
    \end{align}
    Replacing in~\cref{eq:traceinner} we obtain the stated result. The case of pure states in a direct consequence of this more general result.
\end{proof}

\end{document}